\documentclass[a4paper]{article}

\usepackage{authblk}
\author{David Monniaux}
\affil{CNRS / VERIMAG}

\title{$\oracleclass{\NP}{\sharpP} = \existsclass{\PP}$ and other remarks about maximized counting}

\newcommand{\complexityclass}[1]{\mathsf{#1}}
\newcommand{\problem}[1]{\mathrm{#1}}
\newcommand{\Pclass}{\complexityclass{P}}
\newcommand{\nondetclass}[1]{\complexityclass{N}{#1}}
\newcommand{\NP}{\nondetclass{\Pclass}}
\newcommand{\majorityclass}[1]{\complexityclass{P}{#1}}
\newcommand{\PP}{\majorityclass{\Pclass}}
\newcommand{\counting}[1]{\#{#1}}
\newcommand{\existsclass}[1]{\exists{#1}}
\newcommand{\sharpP}{\counting{\Pclass}}
\newcommand{\oracleclass}[2]{{#1}^{#2}}
\newcommand{\limitedoracleclass}[3]{{#1}^{#2[#3]}}
\newcommand{\SAT}{\problem{SAT}}

\newcommand{\ve}[1]{\mathbf{#1}}
\newcommand{\countmodels}[1]{\counting\SAT(#1)}
\newcommand{\countset}[1]{\#{\left\{#1\right\}}}
\newcommand{\size}[1]{\left|#1\right|}

\usepackage{amsthm}
\theoremstyle{plain}
\newtheorem{lemma}{Lemma}
\newtheorem{theorem}{Theorem}
\newtheorem{corollary}{Corollary}
\theoremstyle{remark}
\newtheorem{remark}{Remark}
\theoremstyle{definition}
\newtheorem{definition}{Definition}

\usepackage[backend=bibtex]{biblatex}
\bibliography{counting_complexity}

\usepackage{hyperref}

\date{February 23, 2022}

\begin{document}
\maketitle

We consider the following decision problem $\problem{DMAX\#SAT}$, and generalizations thereof: given a quantifier-free propositional formula $F(\ve{x},\ve{y})$, where $\ve{x},\ve{y}$ are tuples of variables, and a bound $B$, determine if there is $\ve{x}$ such that $\countset{\ve{y} \mid F(\ve{x},\ve{y})} \geq B$.
This is the decision version of the problem of $\problem{MAX\#SAT}$: finding $\ve{x}$ and $B$ for maximal~$B$.

\begin{theorem}
  $\problem{DMAX\#SAT}$ is $\existsclass{\PP}$-complete.
\end{theorem}

\begin{proof}
  It is in $\existsclass{\PP}$: it is well-known that taking $(F,B)$ as input and checking if $\countmodels{F} \geq B$ is in~$\PP$.
  
  Take a problem in $\existsclass{\PP}$.
  It can be reformulated as: take input $x$,
  choose nondeterministic bits $y$, construct a formula $F(x,y)$ with $N(x,y)$ variables, and check that
  $\#\{z \in \{0,1\}^{N(x,y)} \mid F(x,y)(z)\} \geq 2^{N(x,y)-1}$.
  The condition $F(x,y)(z)$ can also be reformulated as $\exists z'~G(x,y,z,z')$ where $G$ simulates the action of the Turing machine that produces $F$ (if necessary by using temporary values in $z'$) then the semantics of the formula over~$z$.
  The result follows.
\end{proof}

\textcite[theorem 4.1 (ii)]{DBLP:journals/jacm/Toran91} showed that $\existsclass{\PP} = \oracleclass{\NP}{\sharpP}$; actually, a generalization of this.
However, prior to becoming aware of that result, we had worked out another proof, which we present here.

The following gadgets enables us to transform multiple equality tests over model counts $\countmodels{F_1}=C_1 \land \dots \land \countmodels{F_m}=C_m$ into a single equality test over model counts.

\begin{definition}
Let $F$ and $G$ be two quantifier-free propositional formulas with $m$ and $n$ variables respectively. Without loss of generality, we assume these variables to be $x_1,\dots,x_m$ and $x_1,\dots,x_n$.
Let $\phi_2^{m,n}(F,G)$ be the following formula over $m+n+2$ variables:
\begin{equation}
  \left(F(x_1,\dots,x_m) \land \neg{x_{m+1}} \land \dots \land \neg{x_{m+n+2}}\right) \lor
  \left(G(x_1,\dots,x_m) \land x_{m+1} \right)
\end{equation}
By $\size{F}$ we denote the size of a formula as the number of its Boolean operators, and by $\countmodels{F}$ we denote the number of its models. 
\end{definition}

\begin{lemma}
  $\size{\phi_2^{m,n}(F,G)} = \size{F} + \size{G} + n+ 3$. Furthermore,
  $\countmodels{F}$ and $\countmodels{G}$ are respectively the remainder and quotient of $\countmodels{\phi_2^{m,n}(F,G)}$ by $2^{n+1}$.
\end{lemma}


\begin{definition}
Let $F_0,\dots$ be propositional formulas with $n$ variables.
Let $\phi_1^n(F_0)=F_0$, $\phi_{k+1}^n(F_0,\dots,F_k) = \phi_2^{kn+2(k-1),n}(\phi_k(F_0,\dots,F_{k-1}), F_k)$.
\end{definition}

\begin{lemma}\label{ref:count_base_decomposition}
  $\size{\phi_k^n(F_0,\dots,F_{k-1})} = \sum_i \size{F_i} + (k-1)(n+3)$.
Furthermore, $\countmodels{F_i}$ is the digit of order $i$ (starting with $i=0$) of the decomposition of $\countmodels{\phi_2^{m,n}(F,G)}$ in base $2^{n+1}$.
\end{lemma}

The following gadget will be used to add a number of models to an existing formula:
\begin{definition}
Let $M^n_c(x_0,\dots,x_{n-1})$, where $0 \leq c \leq 2^n$, be the formula that specifies that $\sum_i 2^i x_i \leq c$.
\end{definition}

\begin{lemma}
  $\size{M^n_c}$ is linear in $n$, and $\countmodels{M^n_c}=c$.
\end{lemma}

The following gadget will be used to turn an equality test on the number of models of a formula into a ``greater than or equality'' inequality test on the number of models of another formula:

\begin{definition}
Let $F$ be a propositional formula over $n$ variables, and $0 \leq \Delta \leq 2{n-1}$.
Let $\psi^n(F)$ be the formula over $2n+1$ variables
\begin{equation}
  F(x_1,\dots,x_n) \land \left(
  \left(\neg F(x_{n+1},\dots,x_{2n}) \land \neg x_{2n+1} \right) \lor
  \left(M^n_{2\Delta}(x_{n+1},\dots,x_{2n}) \land x_{2n+1}\right)\right)
\end{equation}
Let $K^n_\Delta$ be the polynomial $K^n_\Delta(X) = X(2^n-X+2\Delta)$.
\end{definition}

\begin{lemma}\label{lem:model_count_eq_to_ineq}
  $\size{\psi^n(F)}$ has size linear in $\size{F}$, and
  $\countmodels{\psi^n(F)} = K^n_\Delta(\countmodels{F})$.
  Furthermore, $K^n_\Delta(\countmodels{F}) \geq K^n_\Delta(2^{n-1}+\Delta)$
  if and only if $\countmodels{F}=2^{n-1}+\Delta$.
\end{lemma}

\begin{theorem}
  $\existsclass{\PP} = \limitedoracleclass{\NP}{\PP}{1} = \oracleclass{\NP}{\PP} = \oracleclass{\NP}{\sharpP}$.
\end{theorem}

\begin{proof}
  Inclusions from left to right are trivial. We shall now prove that $\oracleclass{\NP}{\sharpP}$ is included in $\existsclass{\PP}$ by transforming a nondeterministic Turing machine $M$ deciding a problem $D(x)$ in time $P(\size{x})$ with a $\counting\SAT$ oracle into an equivalent decision procedure in $\existsclass{\PP}$.

  We proceed in steps:
  \begin{enumerate}
  \item $M$ calls the oracle at most $P(\size{x})$ times, over formulas of $P(\size{x})$ variables, and the outputs of each oracle call may be used for computing the inputs to further oracle calls.
    Instead, we transform the machine to nondeterministically choose all inputs $F_1, \dots, F_{P(\size{x})}$ and candidate outputs to the oracle calls ($\counting\SAT$), and then only at the end we verify that the candidate outputs $C_1,\dots,C_{P(\size{x})}$ match the real outputs $\countmodels{F_1},\dots,\countmodels{F_{P(\size{x})}}$ (we reject otherwise).
  \item We replace these calls by a single call to $\countmodels{\phi^{P(\size{x})}_{P(\size{x})}(F_1, \dots, F_{P(\size{x})})}$, and a single verification that the output $V$ of this call matches the candidate outputs $C_1,\dots,C_{P(\size{x})}$ according to the decomposition in Lemma~\ref{ref:count_base_decomposition}.
      \item We replace this call and equality test $\countmodels{\phi^{P(\size{x})}_{P(\size{x})}(F_1, \dots, F_{P(\size{x})})} = Y$ by a single call to $\counting\SAT$ and an inequality test as follows.
        Let $n$ be the number of variables of $\phi^{P(\size{x})}_{P(\size{x})}(F_1, \dots, F_{P(\size{x})})$. Two cases:
        \begin{itemize}
        \item $Y \geq 2^{n-1}$, then write $Y = 2^{n-1} + \Delta$.
          Then we replace the equality test by an inequality test
          $\countmodels{\psi^n(\phi^{P(\size{x})}_{P(\size{x})}(F_1, \dots, F_{P(\size{x})}))} \geq K^n_\Delta$, according to Lemma~\ref{lem:model_count_eq_to_ineq}.
        \item $Y < 2^{n-1}$, then write $Y = 2^{n-1} - \Delta$.
          Then we replace the equality test by an inequality test
          $\countmodels{\psi^n(\neg \phi^{P(\size{x})}_{P(\size{x})}(F_1, \dots, F_{P(\size{x})}))} \geq K^n_\Delta$, according to Lemma~\ref{lem:model_count_eq_to_ineq}.
        \end{itemize}
      \item We have thus reduced the procedure to the nondeterministic (possiblyfailing) construction of a pair $(G,B)$ followed by a test $\countmodels{G} \geq B$.
        This test is well-known to be complete for~$\PP$.
  \end{enumerate}
\end{proof}

\begin{corollary}
  $\problem{DMAX\#SAT}$ is hard for the polynomial hierarchy.
\end{corollary}

\begin{proof}
  Obviously, $\oracleclass{\Pclass}{\sharpP} \subseteq \oracleclass{\NP}{\sharpP}$, and the former class is hard for the polynomial hierarchy by Toda's theorem.
\end{proof}

\begin{remark}
  The above lemmas and theorems consider quantifier-free propositional formulas over a number of free variables.
  In fact, we can use arbitrary predicates with a certain number of free variables: for instance, take inputs $x_1,\dots,x_n$, and return true or false depending on whether a certain polynomial-time nondeterministic Turing machine parameterized by $x_1,\dots,x_n$ has an accepting run or not.
  Equivalently, we could consider predicates of the form $\exists z_1,\dots,z_p~ F(x_1,\dots,x_n, z_1,\dots,z_p)$.
  We then obtain the result
  $\existsclass{\majorityclass{\NP}} = \oracleclass{\NP}{\counting{\NP}}$.

  By going with classes of predicates arbitrarily high in the counting hierarchy, we obtain the general theorem \cite[theorem 4.1 (ii)]{DBLP:journals/jacm/Toran91}: for any class $K$ is the counting hierarchy,
  $\existsclass{\majorityclass{K}} = \limitedoracleclass{\NP}{\majorityclass{K}}{1} = \oracleclass{\NP}{\majorityclass{K}} = \oracleclass{\NP}{\counting{K}}$.
  
  This result is not limited to the counting hierarchy, it is appropriate for classes of predicates stable by certain operations (conjunction, disjunction, conjunction with extra propositional inputs\dots).
\end{remark}

\printbibliography
\end{document}